\documentclass[5p,twocolumn]{elsarticle-arxiv}
\usepackage{amsmath}

\usepackage{shuffle}
\usepackage{bm}
\usepackage{amssymb} 
\usepackage{mathtools}
\usepackage{url}
\usepackage{rotating}

\newcommand{\Nat}{{\mathbb{N}}}
\newcommand{\Reals}{{\mathbb{R}}}

\newcommand{\tta}{{\mathtt{a}}}
\newcommand{\ttb}{{\mathtt{b}}}
\newcommand{\ttc}{{\mathtt{c}}}
\newcommand{\size}[1]{\mathopen{\mid}#1\mathclose{\mid}}
\newcommand{\subword}{\preccurlyeq}
\newcommand{\supword}{\succcurlyeq}
\newcommand{\egdef}{\stackrel{\mbox{\begin{scriptsize}def\end{scriptsize}}}{=}}
\newcommand{\equivdef}{\stackrel{\mbox{\begin{scriptsize}def\end{scriptsize}}}{\Leftrightarrow}}

\newcommand{\Bcal}{{\mathcal{B}}}
\newcommand{\obracew}[2]{{\overset{#2}{\overbrace{#1}}}}
\newcommand{\oneset}[1]{\left\{\mathinner{#1}\right\}}
\newcommand{\smallset}[1]{\left\{\mathinner{#1}\right\}}
\newcommand{\floor}[1]{\left\lfloor\mathinner{#1} \right\rfloor}

\newtheorem{theorem}{Theorem}[section]

\newtheorem{lemma}[theorem]{Lemma}
\newtheorem{proposition}[theorem]{Proposition}
\newproof{proof}{Proof}

\begin{document}

\begin{frontmatter}
\title{On the index of Simon's congruence for piecewise testability}
\author[lsv,cmi]{P. Karandikar\fnref{tcs,anr}}
\author[fmi]{M. Kuf\-leitner\fnref{dfg}}
\author[lsv]{Ph. Schnoebelen\fnref{anr}}
\address[lsv]{Lab.\  Specification \& Verification, CNRS UMR 8643 \& ENS Cachan, France}
\address[cmi]{Chennai Mathematical Institute, Chennai, India}
\address[fmi]{Institut f{\"u}r Formale Methoden der Informatik, University of Stuttgart, Germany}
\fntext[tcs]{Partially supported by Tata Consultancy Services.}
\fntext[anr]{Supported by ANR grant 11-BS02-001-01.}
\fntext[dfg]{Supported by DFG grant \mbox{DI~435/5-2}.}

\begin{abstract}
Simon's congruence, denoted $\sim_n$, relates words having the same
subwords of length up to $n$. We show that, over a $k$-letter
alphabet, the number of words modulo $\sim_n$ is in
$2^{\Theta(n^{k-1}\log\, n)}$.
\end{abstract}

\end{frontmatter}

\section{Introduction}
\label{sec-intro}

Piecewise testable languages, introduced by Imre Simon in the 1970s,
are a family of star-free regular languages that are definable by the
presence and absence of given (scattered) subwords~\cite{simon75,sakarovitch83,pin86}. 
Formally, a language
$L\subseteq A^*$ is $n$-piecewise testable if $x\in L$ and $x\sim_n y$
imply $y\in L$, where $x\sim_n y$ $\equivdef$ $x$ and $y$ have the
same subwords of length at most $n$ (see next section for
all definitions missing in this introduction). Piecewise testable
languages are important because they are the languages defined by
$\Bcal\Sigma_1$ formulae, a simple fragment of first-order logic that
is prominent in database queries~\cite{DGK-ijfcs08}. They also occur in learning
theory~\cite{kontorovich2008}, computational
linguistics~\cite{rogers2010}, etc.

It is easy to see that $\sim_n$ is a congruence with finite index and
Sakarovitch and Simon raised the question of how to better
characterize or evaluate this number~\cite[p.~110]{sakarovitch83}. Let
us write $C_k(n)$ for the number of $\sim_n$ classes over $k$ letters, i.e.,
when $\size{A}=k$. It is clear that $C_k(n)\geq k^n$ since two
 words $x,y\in A^{\leq n}$ (i.e., of length at most $n$) are related
by $\sim_n$ only if they are equal. In fact, this reasoning gives
\begin{gather}
\label{eq-naive-lower-bound}
C_k(n) \geq k^n + k^{n-1} + \cdots + k + 1 = \frac{k^{n+1} - 1}{k-1}
\end{gather}
(assuming $k\neq 1$).
On the other hand, any congruence class in $\sim_n$ is completely
characterized by a set of subwords in $A^{\leq n}$, hence
\begin{gather}
\label{eq-naive-upper-bound}
C_k(n)\leq 2^{\frac{k^{n+1} - 1}{k-1}}
\:.
\end{gather}
Estimating the size of $C_k(n)$ has applications in
descriptive complexity, for example for estimating the number of
$n$-piecewise testable languages (over a given alphabet), or for
bounding the size of canonical automata for $n$-piecewise testable
languages~\cite{czerwinski2013,klima2013,place2013}.

Unfortunately the above bounds, summarized as $k^n\leq C_k(n)\leq
2^{k^{n+1}}$, leave a large (``exponential'') gap and it is not clear
towards which side is the actual value leaning.\footnote{Comparing the bounds
from Eqs.~\eqref{eq-naive-lower-bound} and \eqref{eq-naive-upper-bound} with actual values  does not bring much light here
since the magnitude of $C_k(n)$ makes it hard to compute beyond some
very small values of $k$ and $n$, see Table~\ref{table-C}.}
Eq.~\eqref{eq-naive-lower-bound} gives a lower bound that is obviously
very naive since it only counts the simplest classes. On the other
hand, Eq.~\eqref{eq-naive-upper-bound} too makes wide simplifications
since not every subset of $A^{\leq n}$ corresponds to a
congruence class. For
example, if $\tta\tta$ and $\ttb\ttb$ are subwords of some $x$ then
necessarily $x$ also has $\tta\ttb$ or $\ttb\tta$ among its length 2 subwords.

Since the question of estimating $C_k(n)$ was raised
in~\cite{sakarovitch83} (and to the best of our knowledge) no progress has been made on the
question, until K\'atai-Urb\'an et al.\  proved the following bounds:
\begin{theorem}[\citet{KPPPS2012}] For all $k>1$,
\label{theo-kppps}
\begin{alignat*}{2}
\frac{k^n}{3^{n^2}}  \log \, k& \leq \log \, C_k(n) < 3^n k^n  \log \, k
&&\quad\text{if $n$ is even,}
\\
\frac{k^n}{3^{n^2}}& < \log \, C_k(n) < 3^n k^n
&&\quad\text{if $n$ is odd.}
\end{alignat*}
\end{theorem}
The proof is based on two reductions, one showing $C_{k+\ell}(n+2)\geq C_k^{\ell+2}(n)$ for
proving lower bounds, and one showing $C_k(n+2)\leq (k+1)^{2k}
C_k^{2k-1}(n)$ for proving upper bounds.
For fixed $n$, Theorem~\ref{theo-kppps} allows to estimate the
asymptotic value of $\log\, C_k(n)$ as a function of $k$: it is in
$\Theta(k^n)$ or $\Theta(k^n\log\, k)$ depending on the parity of $n$. However,
these bounds do not say how, for fixed $k$, $C_k(n)$ grows as a
function of $n$, which is a more natural question in settings where
the alphabet is fixed, and where $n$ comes from, e.g., the number of
variables in a $\Bcal\Sigma_1$ formula. In particular, the lower bound
is useless for $n\geq k$ since in this case $k^n/3^{n^2}<1$.

\paragraph{Our contribution}
In this article, we provide the following bounds:
\begin{theorem} For all $k,n>1$,
\label{theo-summary}
\begin{align*}
\Bigl(\frac{n}{k}\Bigr)^{k-1} \log_2 \left(\frac{n}{k}\right)
&<
\log_2 C_k(n) \\
&<
k\left(\frac{n+2k-3}{k-1}\right)^{k-1} \log_2 n \log_2 k
\:.
\end{align*}
\end{theorem}
Thus, for fixed $k$, $\log\, C_k(n)$ is in $\Theta(n^{k-1}\log\, n)$.
Compared with Theorem~\ref{theo-kppps}, our bounds are much tighter for
fixed $k$ (and much wider for fixed $n$).

The proof of Theorem~\ref{theo-summary} relies on two new
reductions that allows us to relate $C_k(n)$ with $C_{k-1}$
instead of relating it with $C_k(n-2)$ as in~\cite{KPPPS2012}. The
article is organized as follows. Section~\ref{sec-basics} recalls the
necessary notations and definitions; the lower bound is proved in
Section~\ref{sec-lower-bound} while the upper bound is proved in
Section~\ref{sec-upper-bound}. An appendix lists the exact values of $C_k(n)$
for small $n$ and $k$ that we managed to compute.

\section{Basics}
\label{sec-basics}

We consider words $x,y,w,\ldots$ over a finite $k$-letter alphabet
$A_k=\{\tta_1,\ldots,\tta_k\}$
sometimes written more simply
$A=\{\tta,\ttb,\ldots\}$.
The empty word is denoted $\epsilon$,
concatenation is denoted multiplicatively. Given a word $x\in A^*$ and
a letter $\tta\in A$, we write $\size{x}$ and $\size{x}_\tta$ for,
respectively, the length of $x$, and the number of occurrences of
$\tta$ in $x$.

We write $x\subword y$ to denote that a word $x$ is a \emph{subsequence} of
$y$, also called a (scattered) \emph{subword}. Formally, $x\subword y$ iff
$x=x_1\cdots x_\ell$ and there are words $y_0,y_1,\ldots,y_\ell$ such
that $y=y_0  x_1  y_1\cdots x_\ell  y_\ell$. It is well-known that
$\subword$ is a partial ordering and a monoid precongruence.

For any $n\in\Nat$, we write $x\sim_n y$ when $x$ and $y$ have the
same subwords of length $\leq n$. For example
$x\egdef\tta\ttb\tta\ttc\ttb \sim_2
y\egdef\ttb\tta\tta\tta\ttc\ttb\ttb$ since both words have
$\{\epsilon,\tta,\ttb,\ttc,\tta\tta,\tta\ttb,\tta\ttc,\ttb\tta,\ttb\ttb,\ttb\ttc,\ttc\ttb\}$
as subwords of length $\leq 2$. However $x\not\sim_3 y$ since
$x\supword\tta\ttb\tta\not\subword y$. Note that $\sim_0 \,\supseteq\,
\sim_1 \,\supseteq\, \sim_2 \,\supseteq\, \cdots$, and that $x\sim_0 y$ holds
trivially. It is well-known (and easy to see) that each $\sim_n$ is a
congruence since the subwords of some $x y$ are the concatenations of
a subword of $x$ and a subword of $y$. Simon defined a \emph{piecewise
  testable} language as any $L\subseteq A^*$ that is closed by
$\sim_n$ for some $n$~\cite{simon75}. These are exactly the languages
definable by $\Bcal\Sigma_1(<,\tta,\ttb,\ldots)$ formulae~\cite{DGK-ijfcs08}, i.e., by
Boolean combinations of existential first-order formulae with monadic
predicates of the form $\tta(i)$, stating that the $i$-th letter of a word is
$\tta$. For example, $L=A^* \tta A^* \ttb A^*=\{x\in A^*~|~
\tta\ttb\subword x\}$ is definable with the following $\Sigma_1$
formula:
\[
\exists i:\exists j:i<j\land \tta(i)\land\ttb(j)
\:.
\]

\paragraph{The index of $\sim_n$}
Since there are only finitely many words of length $\leq n$, the
congruence $\sim_n$ partitions $A_k^*$ in finitely many classes, and
we write $C_k(n)$ for the number of such classes, i.e., the cardinal
of $A_k^*\,/\!\sim_n$.

The following is easy to see:
\begin{xalignat}{3}
\label{eq-easy}
C_1(n) & = n+1\:,
&
C_k(0) & = 1\:,
&
C_k(1) & = 2^k\:.
\end{xalignat}
Indeed, for words over a single letter $\tta$, $x\sim_n y$ iff
$\size{x}=\size{y}< n$ or $\size{x}\geq n\leq\size{y}$, hence the
first equality. The second equality restates that $\sim_0$ is trivial,
as noted above. For the third equality, one notes that $x\sim_1 y$
if, and only if, the same set of letters is occurring in $x$ and $y$,
and that there are $2^k$ such sets of occurring letters.


\section{Lower bound}
\label{sec-lower-bound}

The first half of Theorem~\ref{theo-summary} is proved by first
establishing a combinatorial inequality on the $C_k(n)$'s
(Proposition~\ref{prop-Ckn-geq}) and then using it to derive
Proposition~\ref{prop-lower-bound}.
\\

Consider two words $x,y\in A^*$ and a letter $a\in A$.
\begin{lemma}
\label{lem-simn-nbofas}
If $x \sim_n y$, then $\min (\size{x}_a,n) =  \min
(\size{y}_a,n)$.
\end{lemma}
\begin{proof}[Sketch]
If $\size{x}_a=p<n$ then $a^p\subword x\not\supword a^{p+1}$. From
$x\sim_n y$ we deduce $a^p\subword y\not\supword a^{p+1}$, hence
$\size{y}_a=p$.
\qed
\end{proof}
Fix now $k\geq 2$, let $A=A_k=\{\tta_1,\ldots,\tta_k\}$ and assume
$x\sim_n y$. If $\size{x}_{\tta_k}=p<n$, then $x$ is some $x_0
\tta_k x_1 \cdots \tta_k x_p$ with $x_i\in A^*_{k-1}$ for
$i=0,\ldots,p$. By
Lemma~\ref{lem-simn-nbofas}, $y$ too is some $y_0 \tta_k
y_1 \cdots \tta_k y_p$ with $y_i\in A^*_{k-1}$.
\begin{lemma}
\label{lem-simn-p-xiyi}
$x_i \sim_{n-p} y_i$ for all $i=0,\ldots,p$.
\end{lemma}
\begin{proof}
Suppose $w\subword x_i$ and $\size{w}\leq n-p$. Let $w'\egdef \tta_k^i
w \tta_k^{p-i}$. Clearly $w'\subword x$ and thus $w'\subword y$ since
$x\sim_n y$ and $\size{w'}\leq n$. Now $w'=\tta_k^i w
\tta_k^{p-i}\subword y$ entails $w\subword y_i$.

With a symmetric reasoning we show that every subword of $y_i$ having
length $\leq n-p$ is a subword of $x_i$ and we conclude $x_i\sim_{n-p}
y_i$.
\qed
\end{proof}

\begin{proposition}
\label{prop-Ckn-geq}
For $k\geq 2$, $C_{k}(n) \geq \sum_{p=0}^n C_{k-1}^{p+1}(n-p)$.
\end{proposition}
\begin{proof}
For words $x = x_0 \tta_k x_1 \ldots x_{p-1} \tta_k x_p$ with exactly
$p<n$ occurrences of $\tta_k$, we have $C_{k-1}(n-p)$ possible choices
of $\sim_{n-p}$ equivalence classes for each $x_i$ ($i=0,\ldots,p$).
By Lemma~\ref{lem-simn-p-xiyi} all such choices will result in
$\not\sim_n$ words, hence there are exactly $C^{p+1}_{k-1}(n-p)$
classes of words with $p<n$ occurrences of $\tta_k$. By
Lemma~\ref{lem-simn-nbofas}, these classes are disjoint for different
values of $p$, hence we can add the $C_{k-1}^{p+1}(n-p)$'s.
There remain words with $p\geq n$ occurrences of $\tta_k$, accounting
for at least $1$, i.e., $C_{k-1}^{n+1}(0)$, additional class.
\qed
\end{proof}

\begin{proposition}\label{prp:lower}\label{prop-lower-bound}
  For all $k,n>0$:
  \begin{equation}\label{eq-lower-bound}
    \log_2 C_k(n) > \left(\frac{n}{k}\right)^{k-1} \log_2 \left(\frac{n}{k}\right).
  \end{equation}
\end{proposition}
\begin{proof}
Eq.~\eqref{eq-lower-bound} holds trivially when 
$\log_2(\frac{n}{k})\leq 0$. Hence there only remains to consider the
cases where $n>k$. We reason by induction on $k$.
For $k=1$, Eq.~\eqref{eq-easy} gives $\log_2 C_1(n) = \log_2 (n + 1) > \log_2 n =
\left(\frac{n}{1}\right)^{0} \log_2 \left(\frac{n}{1}\right)$.
For the inductive case, Proposition~\ref{prop-Ckn-geq} yields $C_{k+1}(n) \geq
C_{k}^{p+1}(n-p)$ for all $p \in \oneset{0,\ldots,n}$. For $p = \floor{\frac{n}{k+1}}$ this yields
\begin{align*}
  \log_2 C_{k+1}(n) &\geq (p+1) \log_2 C_k(n-p)
  \\ &> (p+1) \left( \frac{n-p}{k} \right)^{k-1} \log_2 \left( \frac{n-p}{k} \right)
\\
\shortintertext{by ind.\ hyp., noting that $n-p>0$,}
 &\geq \frac{n}{k+1} \left( \frac{n}{k+1} \right)^{k-1} \log_2 \left( \frac{n}{k+1} \right)
  \\
\shortintertext{since $\frac{n-p}{k} \geq \frac{n}{k+1}\geq 1$,}
 &= \left( \frac{n}{k+1} \right)^{k} \log_2 \left( \frac{n}{k+1} \right)
\end{align*}
as desired.
\qed
\end{proof}


\section{Upper bound}
\label{sec-upper-bound}

The second half of Theorem~\ref{theo-summary} is again by
establishing a combinatorial inequality on the $C_k(n)$'s
(Proposition~\ref{prop-Ckn-comb-ineq})
and then using it to derive Proposition~\ref{prop-upper-bound}.
\\

Fix $k>0$ and consider words in $A_k^*$. We say that a word $x$ is
\emph{rich} if all the $k$ letters of $A_k$ occur in it, and that it
is \emph{poor} otherwise. For $\ell>0$, we further say that $x$
is $\ell$-rich if it can be written as a concatenation of  $\ell$ rich
factors (by extension  ``$x$ is $0$-rich'' means that $x$ is poor).
The \emph{richness} of $x$ is the largest $\ell\in\Nat$ such that $x$ is
$\ell$-rich. Note that $\forall a\in A_k:\size{x}_a\geq \ell$ does not
imply that $x$ is $\ell$-rich. We shall use the following easy result:
\begin{lemma}
\label{lem-rich-and-sim}
If $x_1$ and $x_2$ are respectively $\ell_1$-rich and $\ell_2$-rich,
then $y\sim_n y'$ implies $x_1 y x_2\sim_{\ell_1+n+\ell_2}x_1 y' x_2$.
\end{lemma}
\begin{proof}
A subword $u$ of $x_1 y x_2$ can be decomposed as $u=u_1 v u_2$ where
$u_1$ is the largest prefix of $u$ that is a subword of $x$ and $u_2$
is the largest suffix of the remaining $u_1^{-1} u$ that is a subword
of $x_2$. Thus $v\subword y$ since $u\subword x_1 y x_2$. Now, since
$x_1$ is $\ell_1$-rich, $\size{u_1}\geq \ell_1$ (unless $u$ is too
short), and similarly $\size{u_2}\geq \ell_2$ (unless \ldots). Finally
$\size{v}\leq n$ when $\size{u}\leq \ell_1+n+\ell_2$, and then
$v\subword y'$ since $y\sim_n y'$, entailing $u\subword x_1 y' x_2$. A
symmetrical reasoning shows that subwords of $x_1y'x_2$ of length $\leq
\ell_1+n+\ell_2$ are subwords of $x_1yx_2$ and we are done.
\qed
\end{proof}

The \emph{rich factorization} of $x\in A_k^*$ is the decomposition
$x=x_1 a_1 \cdots x_m a_m y$
 obtained in the following way: if $x$ is poor, we
let $m=0$ and $y=x$; otherwise $x$ is rich, we let $x_1 a_1$ (with
$a_1\in A_k$) be the
shortest  prefix of $x$ that is rich, write $x=x_1 a_1 x'$ and let
$x_2 a_2\ldots x_m a_m
y$ be the rich factorization of the remaining suffix $x'$.
By construction $m$ is the {richness} of $x$.
E.g.,
assuming $k=3$, the following is a rich factorization with $m=2$:
\[
\obracew{\ttb\ttb\tta\tta\tta\ttb\ttb\ttc\ttc\ttc\ttc\tta\tta\ttb\ttb\ttb\tta\tta}{x}
=\obracew{\ttb\ttb\tta\tta\tta\ttb\ttb}{x_1}\cdot\ttc\cdot
\obracew{\ttc\ttc\ttc\tta\tta}{x_2}\cdot\ttb \cdot\obracew{\ttb\ttb\tta\tta}{y}
\]
Note that, by definition, $x_1,\ldots,x_m$ and $y$ are poor.

\begin{lemma}
\label{lem-rich-middle}

Consider two words $x,x'$ of richness $m$ and with rich factorizations $x=x_1 a_1 \ldots x_m a_m y$ and $x'=x'_1 a_1
\ldots x'_m a_m y'$. Suppose that  $y \sim_n y'$ and that
$x_i \sim_{n+1} x'_i$  for all $i=1,\ldots,m$. Then $x \sim_{n+m} x'$.
\end{lemma}
\begin{proof}
By repeatedly using Lemma~\ref{lem-rich-and-sim}, one shows
\begin{align*}
x_1a_1x_2a_2\ldots x_ma_m y
\;\sim_{n+m}\; & x'_1a_1x_2a_2\ldots x_ma_m y
\\
\;\sim_{n+m}\; & x'_1a_1x'_2a_2\ldots x_ma_m y
\\
&\quad\quad\vdots
\\
\;\sim_{n+m}\; & x'_1a_1x'_2a_2\ldots x'_ma_m y
\\
\;\sim_{n+m}\; & x'_1a_1x'_2a_2\ldots x'_ma_m y'
\:,
\end{align*}
using the fact that each factor $x_i a_i$ is rich.
\qed
\end{proof}

\begin{proposition}
\label{prop-Ckn-comb-ineq}
For all $n\geq 0$ and $k\geq 2$,
\begin{align}
\notag
C_k(n) & \leq 1+\sum_{m=0}^{n-1} k^{m+1} \, C_{k-1}^{m}(n-m+1) \,
C_{k-1}(n-m)\:.
\\
\shortintertext{Furthermore, for $k=2$,}
\label{eq-C2n}
C_2(n) & \leq 2 \sum_{m=0}^{2n-1} n^m = 2\frac{n^{2n}-1}{n-1}
\:.
\end{align}
\end{proposition}
\begin{proof}
Consider two words $x,x'$ and their rich factorization $x=x_1 a_1
\ldots x_m a_m y$ and $x'=x'_1 a'_1 \ldots x'_\ell a'_\ell y'$. By
Lemma~\ref{lem-rich-middle} they belong to the same $\sim_n$ class if
$\ell=m$, $y\sim_{n-m}y'$, and $a_i=a'_i$ and $x_i\sim_{n-m+1}x'_i$
for all $i=1,\ldots,m$. Now for every fixed $m$, there are at most
$k^m$ choices for the $a_i$'s, $C_{k-1}^m(n-m+1)$ non-equivalent
choices for the $x_i$'s, $k C_{k-1}(n-m)$ choices for $y$ and a letter
that is missing in it. We only need to consider $m$ varying up to
$n-1$ since all words of richness $\geq n$ are $\sim_n$-equivalent, 
accounting for one additional possible \mbox{$\sim_n$ class}.

For the second inequality, assume that $k=2$ and $A_2=\{\tta,\ttb\}$.
A word $x\in A_2^*$ can be decomposed as a sequence of $m$ non-empty
blocks of the same letter, of the form, e.g., $x=\tta^{\ell_1}
\ttb^{\ell_2} \tta^{\ell_3} \ttb^{\ell_4} \cdots \tta^{\ell_m}$ (this
example assumes that $x$ starts and ends with $\tta$, hence $m$ is
odd). If two words like $x=\tta^{\ell_1} \ttb^{\ell_2} \tta^{\ell_3}
\ttb^{\ell_4} \cdots \tta^{\ell_m}$ and $x'=\tta^{\ell'_1}
\ttb^{\ell'_2} \tta^{\ell'_3} \ttb^{\ell'_4} \cdots \tta^{\ell'_m}$
have the same first letter $\tta$, the same alternation depth $m$,
and have $\min(\ell_i,n)=\min(\ell'_i,n)$ for all $i=1,\ldots,m$, then
they are $\sim_n$-equivalent. For a given $m>0$, there are
 $2$ possibilities for choosing the first letter and  $n^m$ non-equivalent
choices for the $\ell_i$'s. Finally, all words with alternation depths
$m\geq 2n$ are $\sim_n$-equivalent, hence we can restrict our
attention to $1\leq m\leq 2n-1$. The extra summand $2\,n^0$ in
Eq.~\eqref{eq-C2n} accounts for the single class with $m\geq 2n$ and
the single class with $m=0$.
\qed
\end{proof}

\begin{proposition}\label{prp:upper}\label{prop-upper-bound}
  For all $k,n>1$:
  \begin{equation*}
  C_k(n) < 2^{k \left( \frac{n+2k-3}{k-1} \right)^{k-1} \log_2 n \log_2 k}.
  \end{equation*}
\end{proposition}

\begin{proof}
By induction on $k$. For $k=2$,  Eq.~\eqref{eq-C2n} yields:
\begin{align*}
C_2(n)
& \leq 2\frac{n^{2n}-1}{n-1} < n \frac{n^{2n+1}}{1}
\\
\shortintertext{since $n\geq 2$,}
&  = n^{2n+2} = 2^{2(n+1)\log_2 n}
\\
& = 2^{k \left( \frac{n+2k-3}{k-1} \right)^{k-1} \log_2 n \log_2 k}\:.
\end{align*}
For the inductive case,
Proposition~\ref{prop-Ckn-comb-ineq} yields:
  \begin{align*}
    C_{k+1}(n) &\leq 1 + \sum_{m=0}^{n-1} (k+1)^{m+1} C_{k}^m(n-m+1) C_k(n-m)
    \\ &= 1 + (k+1) C_k(n)
    \\ &\quad\ + \sum_{m=1}^{n-1} (k+1)^{m+1} C_{k}^m(n-m+1) C_k(n-m)
    \\ &< (k+1)^n C_k(n) + \sum_{m=1}^{n-1} (k+1)^{n} C_{k}^{m+1}(n-m+1)
    \\
\shortintertext{since $C_k(q)\leq C_k(q+1)$,}
       &< (k+1)^n 2^{k \left( \frac{n+2k-3}{k-1} \right)^{k-1} \log_2 n \log_2 k}
    \\ & \quad\ + \sum_{m=1}^{n-1} (k+1)^n 2^{k (m+1)\left(
      \frac{n-m+2k-2}{k-1} \right)^{k-1} \log_2 n \log_2 k}
    \\
\shortintertext{by ind.\ hyp.,}
       &< (k+1)^n \sum_{m=0}^{n-1} 2^{k (m+1)\left( \frac{n-m+2k-2}{k-1} \right)^{k-1} \log_2 n \log_2 k}.
  \end{align*}
Since $(m+1)\left( \frac{n-m+2k-2}{k-1} \right)^{k-1} \leq \left(
\frac{n+2k-1}{k} \right)^k$ for all $m \in \smallset{0, \ldots,n-1}$
---see~\ref{app-ineq}---, we may proceed with:
  \begin{align*}
    C_{k+1}(n) &< (k+1)^n \sum_{m=0}^{n-1} 2^{k \left( \frac{n+2k-1}{k} \right)^k \log_2 n \log_2 k}
    \\ &= n (k+1)^n 2^{k \left( \frac{n+2k-1}{k} \right)^k \log_2 n \log_2 k}
    \\ &= 2^{\log_2 n + n \log_2(k+1) + k \left( \frac{n+2k-1}{k} \right)^k \log_2 n \log_2 k}
    \\ &< 2^{\left(\log_2 n + n + k \left( \frac{n+2k-1}{k} \right)^k \log_2 n\right) \log_2 (k+1)}
    \\ &< 2^{(k+1) \left( \frac{n+2k-1}{k} \right)^k \log_2 n \log_2(k+1)}
 \end{align*}
since $\log_2 n+n<\left(\frac{n+2k-1}{k}\right)^k\log_2 n$ (see below).
This is the desired bound.

To see that $\log_2 n + n < \left( \frac{n+2k-1}{k} \right)^k \log_2
n$, we use
  \begin{align*}
    \left( \frac{n+2k-1}{k} \right)^k
    &
    > \left(\frac{n}{k}+1\right)^k
    = \sum_{j=0}^k \binom{k}{j} \cdot \left(\frac{n}{k}\right)^j
\\ &=
 1 +
 k\cdot  \left(\frac{n}{k}\right) +
\cdots \geq n+1
\:.
  \end{align*}
  This completes the proof.
  \qed
\end{proof}

By combining the two bounds in Propositions~\ref{prp:lower}
and~\ref{prp:upper} we obtain Theorem~\ref{theo-summary}, implying
that $\log\, C_k(n)$ is in $\Theta(n^{k-1} \log\, n)$ for fixed alphabet size $k$.



\section{Conclusion}
\label{sec-concl}

We proved that, over a fixed $k$-letter alphabet, $C_k(n)$ is
in $2^{\Theta(n^{k-1}\log\, n)}$. This
shows that $C_k(n)$ is not doubly exponential in $n$ as
Eq.~\eqref{eq-naive-upper-bound} and Theorem~\ref{theo-kppps} would allow.
It also is not simply exponential, bounded by a term of the form $2^{{f(k)\cdot
  n^c}}$ where the exponent $c$  does not depend on $k$.

We are still far from having a precise understanding of how $C_k(n)$
behaves and there are obvious directions for improving
Theorem~\ref{theo-summary}. For example, its bounds are not monotonic
in $k$ (while the bounds in Theorem~\ref{theo-kppps} are not monotonic in $n$) and it
only partially uses the combinatorial inequalities given by
Propositions~\ref{prop-Ckn-geq} and~\ref{prop-Ckn-comb-ineq}.


\paragraph{Acknowledgments}
We thank J.\ Berstel, J.-\'E.\ Pin and M.\ Zeitoun for their comments and suggestions.

\bibliographystyle{model1-num-names}
\bibliography{subwords}

\appendix
\section{Additional proofs}
\label{app-ineq}

We prove that $(m+1)\left( \frac{n-m+2k-2}{k-1} \right)^{k-1} \leq
\left(\frac{n+2k-1}{k} \right)^k$ for all $m=0,\ldots,n-1$, an
inequality that was used to establish Proposition~\ref{prp:upper}.
\\

For $k>0$ and $x,y\in\Reals$, let
\begin{align*}
F_k(x) &\egdef \left(\frac{x+2k-1}{k}\right)^k
\:,
\\
G_{k,x}(y) & \egdef (y+1)F_k(x-y+1) = \frac{(y+1)(x-y+2k)^k}{k^k}
\:.
\end{align*}
Let us check that $G_{k,x}\bigl(\frac{k+x}{k+1}\bigr)=F_{k+1}(x)$ for
any $k>0$ and $x\geq 0$:
\begin{align*}
G_{k,x}\left(\frac{k+x}{k+1}\right)
=& \left( \frac{k+x}{k+1} + 1 \right) \frac{1}{k^k}\left(x - \frac{k+x}{k+1} + 2k \right)^k \\
=& \frac{x+2k+1}{k+1}\:\frac{1}{k^k}\left(\frac{kx + 2k^2 + k}{k+1}\right)^k \\
=& \frac{x+2k+1}{k+1}\:\frac{1}{k^k}\left(\frac{k}{k+1}\right)^k\left(x + 2k + 1\right)^k \\
=& \left(\frac{x+2k+1}{k+1}\right)^{k+1} = F_{k+1}(x) \:.
\tag{$\dagger$}
\end{align*}

We now claim that $G_{k,x}(y)\leq F_{k+1}(x)$ for all $y\in[0,x]$. For
$n,k\geq 2$, the claim entails $G_{k-1,n}(m)\leq F_k(m)$, i.e.
$(m+1)\left( \frac{n-m+2k-2}{k-1} \right)^{k-1} \leq \left(
\frac{n+2k-1}{k} \right)^k$, for  $m=0,\ldots,n-1$ as announced.

\begin{proof}[of the claim]
Let $y_{\max}\egdef\frac{k+x}{k+1}$. We prove that $G_{k,x}(y)\leq
G_{k,x}(y_{\max})$ and conclude using Eq.~($\dagger$): $G_{k,x}$ is
well-defined and differentiable over $\Reals$, its derivative is
\begin{align*}
G_{k,x}'(y) &= \frac{(x-y+2k)^k - (y+1)k(x-y+2k)^{k-1}}{k^k} \\
&= \frac{(x-y+2k)^{k-1}}{k^k}\bigl( (x-y+2k) - (y+1)k \bigr) \\
&= \frac{(x-y+2k)^{k-1}}{k^k}\bigl(x + k - y (k+1) \bigr)
\:.
\end{align*}
Thus $G'_{k,x}(y)$ is $0$ for $y=y_{\max}$, is strictly positive for
$0\leq y < y_{\max}$, and strictly negative for $y_{\max}<y\leq x$.
Hence, over $[0,x]$, $G_{k,x}$ reaches its maximum at $y_{\max}$.
\qed
\end{proof}

\section{First values for $C_k(n)$}

We computed the first values of $C_k(n)$ by a brute-force method that
listed all minimal representatives of $\sim_n$ equivalence classes
over a $k$-letter alphabet. Here $x$ is \emph{minimal} if $x\sim_n y$
implies ($\size{x} < \size{y}$ or ($\size{x} = \size{y}$ and  $x\leq_{\text{lex}} y$)).
Every equivalence class has a unique minimal representative.
Note that if a concatenation $xx'$ is minimal
then both $x$ and $x'$ are.
Therefore, when listing the minimal
representatives in order of increasing length, it is possible to stop
when, for some length $\ell$, one finds no minimal representatives. In
that case we know that there cannot exist minimal representatives of
length $>\ell$.

\begin{sidewaystable}
\centering
\renewcommand{\arraystretch}{1.13}
\begin{tabular}{c|r|r|r|r|r|r|r|r||r}
& \multicolumn{1}{|c|}{$k=1$} & \multicolumn{1}{|c|}{$k=2$} & \multicolumn{1}{|c|}{$k=3$} & \multicolumn{1}{|c|}{$k=4$} & \multicolumn{1}{|c|}{$k=5$} & \multicolumn{1}{|c|}{$k=6$} & \multicolumn{1}{|c|}{$k=7$} &  \multicolumn{1}{|c|}{$k=8$} & \multicolumn{1}{|c}{$k$}
\\
\hline
$n=0$ & $1$ & $1$ & $1$ & $1$& $1$& $1$& $1$ & $1$ & $1$
\\
\hline
$n=1$ & $2$ & $4$ & $8$ & $16$ & $32$ & $64$ & $128$ & $256$ & $2^k$
\\
\hline
$n=2$ & $3$ & $16$ & $152$ & $2\,326$ & $52\,132$  & $1\,602\,420$ & $64\,529\,264$ &  $ \geq 173 \cdot 10^7$ &
\\
\hline
$n=3$ & $4$ & $68$ & $5\,312$ & $1\,395\,588$ & $1\,031\,153\,002$                  & $\geq 23 \cdot 10^7$       & & &
\\
\hline
$n=4$ & $5$ & $312$ & $334\,202$ & $\geq 73 \cdot 10^7$                     &  &&&
\\
\hline
$n=5$ & $6$ & $1\,560$ & $38\,450\,477$ &                    &     &    &&
\\
\hline
$n=6$ & $7$ & $8\,528$ & $\geq 39 \cdot 10^7$  &                            &&   &&
\\
\hline
$n=7$ & $8$ & $50\,864$ & &                           &    &&&
\\
\hline
$n=8$ & $9$ & $329\,248$ & &                           &        & &&
\\
\hline
$n=9$ & $10$ & $2\,298\,592$ & &                         &        &   &&
\\
\hline
$n=10$ & $11$ & $17\,203\,264$ & &                         &      &     &&
\\
\hline
$n=11$ & $12$ & $137\,289\,920$ & &                         &      &     &&
\\
\hline
\hline
$n$ & $n+1$ & & &                                   &   &&&&
\end{tabular}
\caption{Computed values for $C_{k}(n)$}
\label{table-C}
\end{sidewaystable}

The
cells left blank in the table were not computed for lack of
memory.

\end{document}